\documentclass[journal]{IEEEtran}

\usepackage{balance}
\usepackage{cite}

\usepackage{subcaption}

\ifCLASSINFOpdf
   \usepackage[pdftex]{graphicx}
   \DeclareGraphicsExtensions{.pdf,.jpeg,.png}
\else
 
\fi

\usepackage{mathrsfs,amsmath}
\usepackage{amssymb}

\usepackage{amsmath}

\usepackage{amsthm,hyperref,enumerate}
\newtheorem{theorem}{Theorem}

\theoremstyle{definition}
\newtheorem{defn}{Definition}

\theoremstyle{remark}
\newtheorem{remark}{Remark}
\newtheorem{lemma}{Lemma}

\theoremstyle{plain}

\usepackage{amsmath}    
  {
      \theoremstyle{plain}
      \newtheorem{assumption}{Assumption}
  }

\theoremstyle{definition}

\usepackage{ragged2e}

\usepackage[amssymb]{SIunits}
\hypersetup{
	colorlinks=true,
	linkcolor=blue,
	filecolor=blue,      
	urlcolor=blue,
}

\begin{document}

\title{Cooperative Pursuit with Multi-Pursuer \\  and One Faster Free-moving Evader}

\author{Xu Fang,
        Chen Wang, 
        Lihua Xie,
        Jie Chen,
\thanks{Xu Fang and Lihua Xie are with the School of Electrical and Electronic Engineering, Nanyang Technological University, 639798, Singapore. (e-mail: fa0001xu@e.ntu.edu.sg; elhxie@ntu.edu.sg).} 
\thanks{Chen Wang is with the Robotics Institute, Carnegie Mellon University, Pittsburgh, PA 15213, USA. (e-mail: chenwang@dr.com).} 
\thanks{Jie Chen is with the School of Electronics and Information Engineering, Tongji University,  Shanghai, 200092, China, and is also with the
Key Laboratory of Intelligent Control and Decision of Complex System,
Beijing Institute of Technology, Beijing, 100081, China. (e-mail: chenjie@bit.edu.cn).
}
}


\maketitle
\begin{abstract}
This paper addresses a multi-pursuer single-evader pursuit-evasion game where the free-moving evader moves faster than the pursuers. Most of the existing works impose constraints on the faster evader such as limited moving area and moving direction. When the faster evader is allowed to move freely without any constraint, the main issues are how to form an encirclement to trap the evader into the capture domain, how to balance between forming an encirclement and approaching the faster evader, and what conditions make the capture possible. In this paper, a
distributed pursuit algorithm is proposed to enable pursuers to form an encirclement and approach the faster evader. An algorithm that balances between forming an encirclement and approaching the faster evader is proposed. Moreover, sufficient capture conditions are derived based on the initial spatial distribution and the speed ratios of the pursuers and the evader.
Simulation and experimental results on ground robots validate the effectiveness and practicability of the proposed method. 
\end{abstract}


\begin{IEEEkeywords}
Pursuit-evasion game, Faster evader, Free-moving, Cooperative pursuit.
\end{IEEEkeywords}

\IEEEpeerreviewmaketitle

\section{Introduction}

\IEEEPARstart{R}{esearchers} across diverse disciplines have been studying collective behaviors such as formation, swarming and schooling in animals. In engineering, scientists want to realize such behaviors in multi-robot applications in terrestrial, space, or oceanic exploration. Technological advances in miniaturizing of sensing, computing and localization devices \cite{zhu2017survey, wang2017ultra, guo2016ultra, wang2018kernel} enable multi-robot of higher complexity to perform cooperative tasks such as surveillance, reconnaissance and rescue search. In these applications, one of the 
most interesting problems is to design a cooperative control scheme for a group of agents (pursuers) to pursuit a target (evader). This kind of problem is called pursuit-evasion game. Both the pursuers and evader try to apply certain strategies, based on the available position and velocity information, to maximize their chance of success in this game \cite{li2016escape, zhang2013near,zhong2018model}.
The classical pursuit-evasion game is a differential game \cite{sac:1965p, Petrosyan:1993p}, which makes use of the Hamilton-Jacobi-Isaacs equation to design pursuit and escape strategies. But it
is not suitable for multi-player games because it will encounter tremendous difficulties in the determination of terminal manifold. For the multi-player game, most of the existing works  \cite{makkapati2018optimal, Zhou:2016gc,yan2018reach} assume that the pursuers have maximum speed no less than that of the evader.  
When the evader has faster speed, it is proved that the faster evader can avoid point capture from any number of pursuers with a lower speed \cite{chernous1976problem}. However, the pursuers can use some tools such as manipulator and net to  achieve capture if the distance between the evader and any pursuer is less than a threshold \cite{Zhou:2016gc,jin2010pursuit,Exarchos:2014cl}, which is called non-zero capture radius.

The existing works related to a faster evader usually impose additional constraints on the faster evader. $(a)$ The faster evader is required to pass between two specified pursuers. A closed-form solution is obtained in terms of the elliptic functions \cite{breakwell1975pursuit, hagedorn1976differential}. The optimal strategy for each pursuer is proposed \cite{Zha:2016hu}, which analyzes
the possibilities of capture and escape;  $(b)$ The faster evader is required to be located at the center of an encirclement formed by the pursuers.  The parallel guidance law is used \cite{Xu:2007wu} to guarantee that the faster evader is always trapped in the capture domain,  which requires that the pursuers be angle-evenly distributed around the evader. The problem becomes challenging when the faster evader is allowed to move freely without any constraint. The winning set and pursuit strategy for the one-pursuer and one-evader game are presented in \cite{bakolas2017finite}.
The capture region is bounded for the one-pursuer and one-evader game \cite{Exarchos:2014cl}.
Therefore, one pursuer can not capture the faster free-moving evader and more pursuers are needed to make capture possible. However, how to design a distributed pursuit algorithm for multi-pursuer?

The pursuit algorithms in \cite{jin2010pursuit,wei2007decentralized} can trap the evader into the capture domain when the pursuers have the same speed as that of the evader. For the case when the evader is of a higher speed than pursuers,
Ramana \cite{ramana2016pursuit} divided the multi-pursuer one-evader game into several independent two-pursuer one-evader subgames, and derived a solution based on a two-pursuer one-evader game with restrictive assumptions on the evader and pursuers.  They try to increase each pursuer' individual success rate of capturing the evader without considering the group success rate of capturing the evader. The problem is that increasing each pursuer' individual success rate of capturing the evader may decrease the group success rate of capturing the evader, and then creates an escapable gap that the evader can escape successfully.  Similar works can be found in \cite{yan2018multiagent,Vieira:2009iu,Chen:2016bc,awheda2016decentralized,cai2009novel}. The pursuit strategies in \cite{ramana2016pursuit,yan2018multiagent,Chen:2016bc,awheda2016decentralized,cai2009novel} are also designed to increase the individual 
success rate of capturing the evader
instead of increasing the group success rate.

The pursuit-evasion game is a team game.
We believe that increasing the group success rate of capturing the evader instead of increasing individual success rate is more reasonable. There is need of an encirclement algorithm that cooperatively increases the group success rate of capturing the evader by trying to trap the faster evader within the capture domain. A Q-learning algorithm is proposed \cite{wang2015research} to form a formation by increasing the group success rate of capturing the faster evader, but it requires discrete state and action spaces. For an encirclement algorithm that cooperatively increases the group success rate of capturing the faster evader, two issues need to be addressed: $(a)$ If there are escapable areas, the pursuers should decrease the escapable areas and try to trap the faster evader into the capture domain. $(b)$ If there is no escapable area, the pursuers should keep the encirclement to avoid creating an escapable area.
Moreover, the pursuers should not only try to form an encirclement, but also try to approach the faster evader. Therefore, there is a potential problem of how to balance between forming an encirclement and approaching the faster evader. 

In addition, in a pursuit-evasion game, it is of interest to know conditions under which the evader can be captured. The capture status and transition condition are given in \cite{Chen:2016bc}. Ramana \cite{ramana2017pursuit} proved that it is impossible for the pursuers to capture the evader in a special 'perfectly encircled formation'. Excluding this special case, the existing capture conditions for capturing a faster evader require that the pursuers be angle-evenly distributed around the evader \cite{Xu:2007wu}.

Therefore, it is obvious when the faster evader is allowed to move freely without any constraint, three problems need to be solved: $(a)$ How to design an encirclement algorithm in continuous space to increase the group success rate of capturing the faster evader? $(b)$ How to balance between forming an encirclement and approaching the evader? $(c)$ Under what conditions, successful capturing of the evader can be guaranteed for random initial distribution of the pursuers.

{This work is based on our previous work \cite{Chen:2013tu}, which presented the notion of overlapping angle and methods for the pursuit-evasion game with a faster evader.} In this paper, {we aim to solve the above three problems in the pursuit-evasion game with a faster free-moving evader.}
A distributed pursuit algorithm is designed for the pursuers with different speeds. The main contributions of this paper are as followings:
\begin{enumerate}
    \item An encirclement algorithm is proposed to trap the faster evader into the capture domain.
    \item An algorithm is proposed to balance between forming an encirclement of the evader and approaching it. The algorithm assigns surrounding task and hunting task to each pursuer.
    \item Sufficient capture conditions are derived that involve the initial spatial distribution and speed ratios of the pursuers and the evader to guarantee capture, regardless of the strategy the faster evader adopts.
\end{enumerate}

This paper is organized as follows: Preliminaries and problem description are given in Section \ref{occupy}. In Section \ref{surround}, an encirclement algorithm is proposed to trap  the faster evader into the capture domain. In Section \ref{pursuit}, an algorithm is presented to balance between forming an encirclement of the evader and approaching it. Section \ref{discrete} provides sufficient capture conditions. Section \ref{simulation} and \ref{experiment} provide the details of simulation and experiment. Finally, we end the paper with some conclusions in Section \ref{conclusion}.

\section{Preliminaries and Problem Description}\label{occupy}

Before describing the problem to be studied in the paper, we introduce some notions including occupied angle,  coverage angle, overlapping occupied angle, and escapable angle.

\subsection{Occupied Angle}

For a group of $n$ pursuers, a local time-varying coordinate system $\Gamma$ centered at the faster evader' time-varying global coordinate $\mathbf{p}_e=(x_e,y_e)$ is shown in Fig. \ref{fig1}. The time-varying global coordinate of pursuer $i$ is denoted by $\mathbf{p}_i=(x_i,y_i)$, $i=1,2,\cdots,n$, and its corresponding coordinate in $\Gamma$ is represented by $\mathbf{p}_{ie}=(x_{ie},y_{ie})=(x_i-x_e,y_i-y_e)$.  In addition, the maximum speeds of pursuer $i$ and the faster evader are denoted by $V_{i}$ and $V_e$, respectively.  The pursuers are governed by a single-integrator kinematic model:
\begin{equation}\label{kinetic}
    \mathbf{\dot {p}}_{i} = \mathbf{v}_i, ||\mathbf{v}_i||_2 \le V_i,
\end{equation}
where $\mathbf{v}_i$ is the control input of the pursuer $i$.

\begin{remark}
Without the information of the moving direction or escape strategy of the evader, 
the pursuers can not predict the motion of the evader
based on a kinematics model of the evader. So, our capture strategy does not assume a kinematics model of the evader.
\end{remark}

\begin{figure}[!t]
\centering
\includegraphics[width=0.8\linewidth]{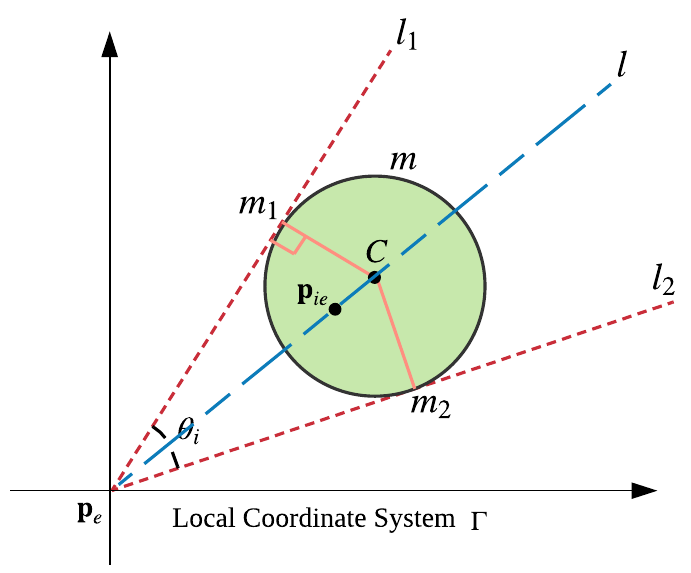}
\caption{The occupied angle}
\label{fig1}
\end{figure}

 (Suppose that both pursuer $i$ and the faster evader begin to move with their maximum speeds along fixed directions at time instant $t_1$. In the local coordinate system $\Gamma(t_1)$, if they meet at a point $\mathbf{m}$ in a finite time $t_1\!+\!T$ in $\Gamma(t_1)$, which is denoted by $\mathbf{m}=(x_m,y_m)$.  The point $\mathbf{m}$ must satisfy the following equation: 
\begin{equation}\label{a1}
\frac{\sqrt{(x_{m}-x_{ie})^2+(y_{m}-y_{ie})^2}}{\sqrt{x_{m}^2+y_{m}^2}} 
=\frac{V_{i} \cdot T}{V_{e} \cdot T}=\lambda_i,
\end{equation}
where $\lambda_i = \frac{V_{i}}{V_e} < 1$ is the speed ratio of pursuer $i$ and the faster evader. Equation \eqref{a1} can be transformed into
\begin{equation}\label{2}
\left(x_{m}-\frac{x_{ie}}{1-\lambda_i^2}\right)^2+\left(y_{m}-\frac{y_{ie}}{1-\lambda_i^2}\right)^2 
=\lambda_i^2\frac{x_{ie}^2+y_{ie}^2}{(1-\lambda_i^2)^2}.
\end{equation}

It is easy to conclude that $\mathbf{m}=(x_m,y_m)$ is on a circle centered at $\mathbf{C}=(\frac{x_{ie}}{1-\lambda_i^2},\frac{y_{ie}}{1-\lambda_i^2})$ with radius $R=\sqrt{x_{ie}^2+y_{ie}^2}\frac{\lambda_i}{1-\lambda_i^2}$ shown in Fig. \ref{fig1}. This circle is known as Apollonius circle \cite{isaacs1999differential}. Two tangent lines $l_1$ and $l_2$ are drawn to this circle from point $\mathbf{p}_{e}$. 
The occupied angle $\theta_i$ is defined as the angle between lines $l_1$ and $l_2$.  If the faster evader moves along a direction between tangent lines $l_1$ and $l_2$ such as $l$, pursuer $i$ can always move in a corresponding direction to capture the evader at a point on the circle. In other words, pursuer $i$ dominates $100\frac{\theta_i}{2\pi}$ percent of the free-moving directions of the evader. The tangent lines $l_1$ and $l_2$ intersect the circle at $\mathbf{m}_1$ and $\mathbf{m}_2$ shown in Fig. \ref{fig1}. According to the sine rule, we have the following relationship:
\begin{equation}\label{3}
\sin(\frac{\theta_i}{2}) = \frac{\mathbf{C} \mathbf{m}_1}{\mathbf{C} \mathbf{p}_e} = \lambda_i,\ \ \lambda_i<1.
\end{equation}

Then the occupied angle can be expressed as
\begin{equation}\label{4}
\theta_i=2\sin^{-1}(\lambda_i).
\end{equation}

The occupied angle $\theta_i$ is only related to the speed ratio of pursuer $i$ and the faster evader. 

\subsection{Coverage Angle, Overlapping Occupied Angle and Escapable Angle}

The local coordinate system $\Gamma$ can be transformed into a local polar coordinate system $L$ with polar coordinate $\mathbf{p}_{ie}(r_i,\alpha_i)$, where
\begin{align}\label{6}
r_i=& \sqrt{x_{ie}^2+y_{ie}^2}, \\
\alpha_i =& \left\{ \begin{array}{lll} \arctan(\frac{y_{ie}}{x_{ie}}), 
& x_{ie}>0,y_{ie}\ge0, \\ 2\pi + \arctan(\frac{y_{ie}}{x_{ie}}),  & x_{ie}>0,y_{ie}<0,\\
\pi + \arctan(\frac{y_{ie}}{x_{ie}}), & x_{ie}<0,\\
\end{array}\right. \label{7}
\end{align}
where $r_i$ is the polar radius and $\alpha_i \in [0,2\pi)$ is the polar angle in $L$. 

A group of $n$ pursuers that disperse counterclockwise in an ascending order with respect to polar angle $\alpha_{i}(\alpha_{i+1}\ge\alpha_{i})$. For adjacent pursuers $\mathbf{p}_{ie}$ and $\mathbf{p}_{(i+1)e}$ shown in Fig. \ref{fig2}, their occupied angles are denoted by $\theta_{i}$ and $\theta_{i+1}$, respectively. Then the coverage angle $\epsilon_{i,i+1}$ of adjacent pursuers is defined as
\begin{subequations}\label{8}
\begin{align}
&\epsilon_{i,i+1}= \alpha_{i+1}-\alpha_i-\frac{\theta_{i+1}+\theta_i}{2},i=1,2,\cdots,n-1,\\
&\epsilon_{n,1}= 2\pi+\alpha_{1}-\alpha_n-\frac{\theta_{1}+\theta_n}{2},i=n.
\end{align}
\end{subequations}

The coverage angle $\epsilon_{i,i+1}$ is divided into two categories: overlapping occupied angle if $\epsilon_{i,i+1} \le 0$ and escapable angle if $\epsilon_{i,i+1} > 0$. 
Both pursuers $\mathbf{p}_{ie}$ and $\mathbf{p}_{(i+1)e}$ can capture the evader if the evader moves along a direction within the overlapping occupied angle such as $\nu$ shown in Fig. \ref{fig2}. But the evader can escape from capture if the evader moves along a direction within the escapable angle because both pursuers $\mathbf{p}_{ie}$ and $\mathbf{p}_{(i+1)e}$ can not occupy this angle.

\begin{figure}[!t]
\centering
\includegraphics[width=0.8\linewidth]{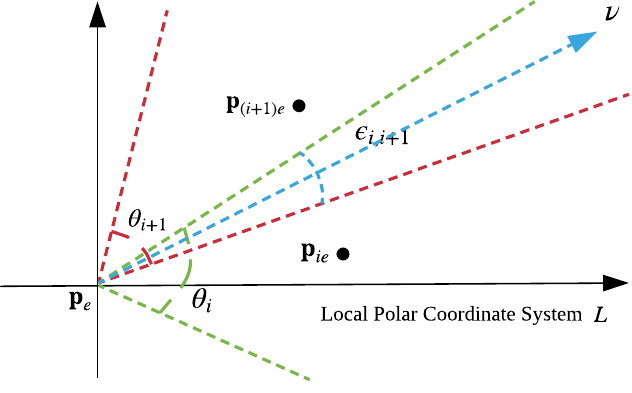}
\caption{The overlapping occupied angle}
\label{fig2}
\end{figure}
 
Therefore, for a group of $n$ pursuers, the group occupied angle $\theta_G$ is defined by 
\begin{equation}\label{9}
\theta_G = \sum_{i=1}^{n}\theta_i +  \sum_{i=1,\epsilon_{i,i+1}\le 0}^{n}\epsilon_{i,i+1}.
\end{equation}

The group occupied angle $\theta_G$ describes that a group of $n$ pursuers dominates $100\frac{\theta_G}{2\pi}$ percent of free-moving directions of the faster evader. 
Then the sum of escapable angles is equal to $\theta_E=2\pi-\theta_G$.
If all pursuers have the same occupied angle $\theta$, the minimum number of pursuers $n_{\min}$ required to occupy all the free-moving directions of the evader is obtained by 
\begin{equation}\label{11}
n_{\min}=\frac{2\pi-\sum_{i=1,\epsilon_{i,i+1}\le0}^{n}\epsilon_{i,i+1}}{\theta} > \frac{2\pi}{\theta}.
\end{equation}

It is easy to understand that the faster evader will be captured by a group of $n$ pursuers if all the free-moving directions are always occupied by the group. In other words, $\theta_G=2\pi$ is always satisfied. Therefore,
The groups success rate $P$ of capturing the faster evader is defined as
\begin{equation}\label{mg}
P= \frac{\theta_G}{2\pi}. 
\end{equation}

\subsection{{Problem Description}}

In this paper, we focus on a pursuit-evasion game with a faster free-moving evader, in a general setting:
\begin{enumerate}
    \item The evader moves faster than a group of pursuers;
    \item The pursuers have heterogeneous maximum speeds;
    \item The moving direction and escape strategy of the faster evader are unknown to the pursuers.
\end{enumerate}

\begin{defn}
The faster evader is said to be captured by a pursuer if the distance between the evader and the pursuer is less than the non-zero capture radius $d_c>0$.
\end{defn}

There are existing works studying the escape strategy of the evader \cite{Zha:2016hu, ramana2016pursuit}. In this paper, we focus on studying pursuit algorithms for the pursuers to capture the faster free-moving evader. More specifically, the problem under investigation is stated as follows:
For a group of $n$ pursuers and one faster free-moving evader, given initial positions $\mathbf{p}_i(0), i\!=\!1,\cdots,n$ and $\mathbf{p}_e(0)$ with $\|\mathbf{p}_i(0) \!-\! \mathbf{p}_e(0)\|_2 > d_c$, find  cooperative velocity control strategies $\mathbf{v}_i, i\!=\!1,\cdots,n$ such that there exist
at least one pursuer $i$ and time instant $t_c>0$ such that $\|\mathbf{p}_i(t_c) \!-\! \mathbf{p}_e(t_c)\| \le d_c$.

Since the moving direction of the evader is unknown to the pursuers, to capture the evader,
the pursuers should not only increase the group success rate $P$ by encircling the evader, but also approach the evader by decreasing their distances to the evader. 
Note that there is a trade-off between
encircling the evader by the pursuers and reducing pursuers' distances to the evader.
In the following sections \ref{surround} and \ref{pursuit}, an encirclement algorithm for increasing $P$, a hunting algorithm for decreasing $r_i$, and their trade-off are presented. 

\begin{figure}[t]
\centering
\includegraphics[width=0.8\linewidth]{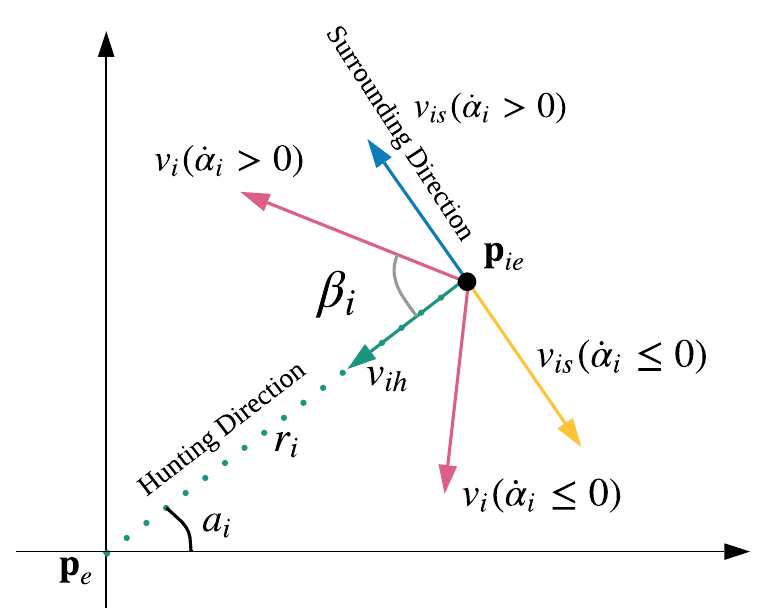}
\caption{The surrounding direction and hunting direction}
\label{fig6}
\end{figure}

\section{Encirclement Algorithm}\label{surround}

Each pursuer has two moving directions: surrounding direction and hunting direction shown in Fig. \ref{fig6}. The hunting direction points to the evader directly, and the surrounding direction is perpendicular to the hunting direction clockwise or counterclockwise. Therefore, the velocity $\mathbf{v}_{i}$ of pursuer $i$ can be projected onto the surrounding direction $\mathbf{v}_{is}$ and hunting direction $\mathbf{v}_{ih}$.  
This section focuses on the problem of increasing the group success rate $P$. Towards this goal, the strategy of the pursuers is to form an encirclement of the evader by cooperation. A distributed encirclement algorithm is proposed for $\mathbf{v}_{is}$ to form an encirclement.

\begin{defn}
The neighbors $\mathbf{p}_{(i+1)e}$ and $\mathbf{p}_{(i-1)e}$ are the 'nearest neighbors' of the pursuer $\mathbf{p}_{ie}$ if their polar angles satisfy $\alpha_{i-1} \le \alpha_{i} \le \alpha_{i+1}$.
\end{defn}

\begin{assumption}
Each pursuer $\mathbf{p}_{ie}$ can access the position information of the faster evader, itself and its two nearest neighbors $\mathbf{p}_{(i\!+\!1)e}$ and $\mathbf{p}_{(i\!-\!1)e}$ with $\alpha_{i\!-\!1} \le \alpha_{i} \le \alpha_{i\!+\!1}$.
\end{assumption}

A group of $n$ pursuers disperse counterclockwise in an ascending order with respect to the polar angle $\alpha_i(\alpha_{i+1}\ge\alpha_i), i=1,\cdots,n$ shown in Fig. \ref{fig5}.  The encirclement algorithm for forming and keeping an encirclement of the evader is designed as
\begin{equation}\label{13}
\dot{\alpha_{i}}=k_i(\epsilon_{i,i+1}-\epsilon_{i-1,i}),i=1,2\cdots,n,
\end{equation}
where $\epsilon_{0,1}=\epsilon_{n,1}$, $\epsilon_{n,n+1}=\epsilon_{n,1}$, and $k_i > 0$ is a surrounding coefficient.  $\dot{\alpha_{i}}>0$ means that pursuer $\mathbf{p}_{ie}$ surrounds the evader counterclockwise shown in Fig. \ref{fig6}, otherwise clockwise. The algorithm \eqref{13} aims to use the sum of the overlapping occupied angles to fill the sum of the escapable angles. Thus, the group occupied angle $\theta_G$ will increase.  

\begin{figure}[t]
\centering
\includegraphics[width=0.73\linewidth]{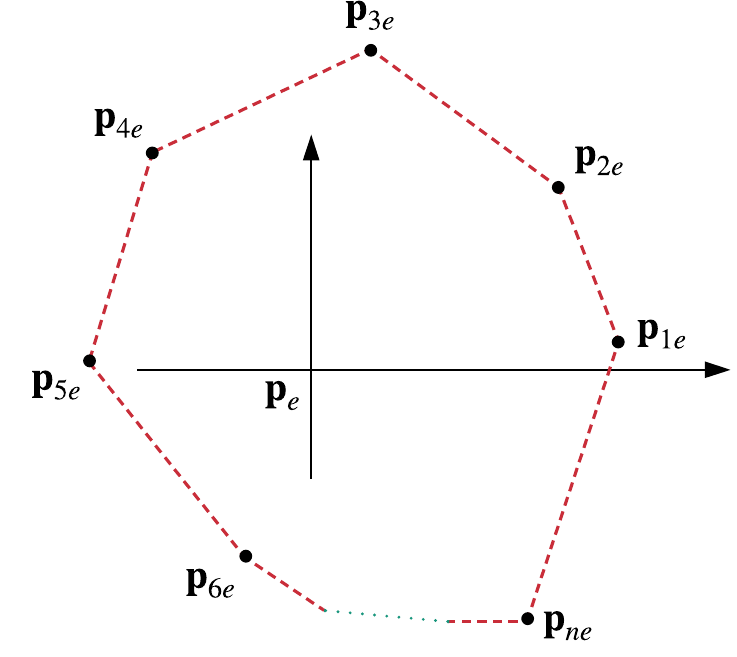}
\caption{The scenario of the pursuit-evasion game}
\label{fig5}
\end{figure}

\begin{theorem}\label{t3}
Consider a local polar coordinate system centered at a static evader' position at time instant $t_0$, the encirclement algorithm \eqref{13} can trap and keep the evader into the union of each pursuer' capture domain if 
\begin{equation}
 2\pi-\sum_{i=1}^{n}\theta_i \le 0.  
\end{equation}
\end{theorem}  
\begin{remark}
The number of the pursuers $n$ and speed ratio $\theta_i$ are the known information and remain unchanged throughout the pursuit-evasion game. Hence, $ 2\pi-\sum\limits_{i=1}^{n}\theta_i$ is constant.
\end{remark}

\begin{proof}
By combining \eqref{8} with \eqref{13}, it yields
\begin{equation}\label{23}
\begin{array}{ll}
\dot{\epsilon}_{i,i+1}&= \dot{\alpha}_{i+1}-\dot{\alpha_i}\\
&=k_{i+1}(\epsilon_{i+1,i+2}-\epsilon_{i,i+1})-k_{i}(\epsilon_{i,i+1}-\epsilon_{i-1,i}) \\
&=-(k_{i+1}+k_i)\epsilon_{i,i+1}+k_{i+1}\epsilon_{i+1,i+2}+k_i\epsilon_{i-1,i} \\
& (i=1,2,\cdots,n),
\end{array} \\
\end{equation}
where $\epsilon_{0,1}=\epsilon_{n,1}$, $\epsilon_{n,n+1}=\epsilon_{n,1}$, $\epsilon_{n+1,n+2}=\epsilon_{1,2}$ and $k_{n+1}=k_1$.
Let vector $\mathbf{\epsilon}=(\epsilon_{1,2},\epsilon_{2,3},\cdots,\epsilon_{n,1})^T \in \mathbb{R}^n$, then, \eqref{23} becomes
\begin{equation}\label{24}
\mathbf{\dot{\epsilon}= -M \epsilon},
\end{equation}
where
\begin{equation}\label{25}
\mathbf{M}\!=\!  \left[\begin{array}{ccccc}
	k_1\!+\! k_2  & -k_2 & 0 & \cdots & -k_1\\
	-k_2  &  k_2\!+\!k_3 & -k_3 & \cdots & 0  \\
	\cdots & \cdots & \cdots & \cdots & \cdots \\
	0 & \cdots & -k_{n-1} & k_{n\!-\!1}\!+\!k_{n} & -k_n \\
	-k_1 & 0 & \cdots & -k_n & k_n\!+\!k_1
	\end{array}\right]
\end{equation}
is symmetrical and the elements denoted by ellipsis are $0$. If the pursuit-evasion game begins at time instant $t_0$, the sum of coverage angles $\sum\limits_{i=1}^{n}\epsilon_{i,i+1}$ satisfies
\begin{equation}\label{26}
\sum_{i=1}^{n}\dot{\epsilon}_{i,i+1}(t)=0, t \ge t_0.
\end{equation}

Then, it yields
\begin{equation}
\sum_{i=1}^{n}{\epsilon}_{i,i+1}(t)= \sum_{i=1}^{n}{\epsilon}_{i,i+1}(t_0), t \ge t_0.
\end{equation}

Substituting \eqref{8} into \eqref{26} gives rise to
\begin{equation}\label{27}
\sum_{i=1}^{n}{\epsilon}_{i,i+1}(t) = 2\pi-\sum_{i=1}^{n}\theta_i,t \ge t_0.   
\end{equation}

The matrix $\mathbf{M}=[m_{ij}]$ satisfies
\begin{equation}\label{p1}
   M(i,j) = \left\{ \begin{array}{lll}
-m_{ij}, &  i \neq j, \\ \sum_{i=1,i\neq j}^{n}m_{ij}, & i=j,
\end{array}\right.
\end{equation}
where $m_{ij}>0$. $\mathbf{M}$ has an eigenvalue $0$ with the associated eigenvector $\mathbf{1}_n$, where $\mathbf{1}_n$ is the $n \times 1$ column vector $[1,1,\cdots,1]^T$, and all other eigenvalues have positive real parts \cite{merris1994laplacian}. In addition, all elements in vector $\epsilon$ will be equal \cite{Wei:2017tu} as $t \rightarrow \infty $
\begin{equation}\label{p2}
\lim_{t\rightarrow \infty}\epsilon_{1,2}(t)=\lim_{t\rightarrow \infty} \epsilon_{2,3}(t)=\cdots=\lim_{t\rightarrow \infty} \epsilon_{n,1}(t).
\end{equation}
Combining \eqref{27} with \eqref{p2} leads to
\begin{equation}\label{28}
\lim_{t \rightarrow \infty}\epsilon_{i,i+1}(t) = \frac{1}{n}(2\pi-\sum_{i=1}^{n}\theta_i),i=1,2,\cdots,n,
\end{equation}
where $\epsilon_{n,n+1}=\epsilon_{n,1}$.

Since $2\pi-\sum\limits_{i=1}^{n}\theta_i \le 0$,  it yields $\lim_{t \rightarrow \infty}\epsilon_{i,i+1}(t)\le 0 (i=1,2,\cdots,n)$ from \eqref{28}. Combining \eqref{9} with \eqref{28} give rise to
\begin{equation}\label{32}
\lim_{t\rightarrow \infty}\theta_G(t) = \sum_{i=1}^{n}\theta_i +  \lim_{t\rightarrow \infty}\sum_{i=1,\epsilon_{i,i+1}(t)\le 0}^{n}\epsilon_{i,i+1}(t) = 2\pi.
\end{equation}

The group occupied angle $\theta_G=2\pi$ means that the evader is trapped into the union of each pursuer' capture domain. The maximum value of $P=1$ \eqref{mg} is achieved.  In addition, when the group occupied angle $\theta_G=2\pi$ and all coverage angles satisfy $\epsilon_{i,i+1}(t)\le0, i=1,2,\cdots,n$, we have $\dot \theta_G=0$, and the group occupied angle will remain as $2\pi$.  Therefore, the algorithm \eqref{13} can trap and keep the evader within the union of each pursuer' capture domain.
\end{proof}

For each pursuer $\mathbf{p}_{ie}(r_i,\alpha_i)$ in local polar coordinate system $L$ shown in Fig. \ref{fig6}, according to the relationship $||\mathbf{v}_{is}||_2=|\dot\alpha_i|r_i, (\dot\alpha_i \in \mathbb{R})$ between linear velocity and angular velocity, the surrounding velocity control $\mathbf{v}_{is}$ under encirclement algorithm \eqref{13} is obtained by
\begin{equation}\label{sp1}
\begin{array}{ll}
      \mathbf{v}_{is}&=\dot\alpha_i r_i  (-\sin\alpha_i,\cos\alpha_i) \\
      &=k_ir_i(\epsilon_{i,i+1}-\epsilon_{i-1,i}) (-\sin\alpha_i,\cos\alpha_i). 
\end{array}
\end{equation}

\section{Velocity Control and Trade-off Algorithm}\label{pursuit}

\subsection{Velocity Control}

The ultimate aim of the pursuers is to capture the faster evader. 
To decrease their distances to the evader, each pursuer should try to move closer to the evader along the hunting direction shown in Fig. \ref{fig6}. For a group of $n$ pursuers, the hunting algorithm is designed as 
\begin{equation}\label{33}
\dot{r}_i=-h_ir_i, i=1,2,\cdots,n,
\end{equation}
where $h_i > 0$ is the hunting coefficient. The hunting velocity along the hunting direction shown in Fig. \ref{fig6} is denoted by $\mathbf{v}_{ih}=(v_{ihx},v_{ihy})=||\mathbf{v}_{ih}||_2(-\cos\alpha_i,-\sin\alpha_i)$. Since $\dot{r}_i=-||\mathbf{v}_{ih}||_2$, it yields
\begin{equation}\label{34}
\mathbf{v}_{ih}=-\dot{r}_i (-\cos\alpha_i,-\sin\alpha_i)=-h_ir_i(\cos\alpha_i,\sin\alpha_i).
\end{equation}

Therefore, for any pursuer $\mathbf{p}_{ie}(\alpha_i,r_i)$, its distributed pursuit velocity control law can be obtained by
\begin{equation}\label{sp2}
\begin{array}{ll}
\mathbf{v}_{i}&=\mathbf{v}_{is}+\mathbf{v}_{ih} \\
& =(-k_ir_i(\epsilon_{i,i+1}-\epsilon_{i-1,i}) \sin\alpha_i-h_ir_i\cos\alpha_i, \\
& k_ir_i(\epsilon_{i,i+1}-\epsilon_{i-1,i}) \cos\alpha_i-h_ir_i\sin\alpha_i).
\end{array}
\end{equation} 

The velocity control law \eqref{sp2} is designed based on \eqref{sp1} and \eqref{34}. Under velocity control law \eqref{sp2},
the group of $n$ pursuers not only try to
trap and keep the evader within the union of each pursuer' capture domain, but also try to approach the evader.

Note that
$\mathbf{v}_{is}$ is for forming an encirclement to increase group success rate $P$  and $\mathbf{v}_{ih}$  is for approaching the evader to decrease $r_i$. Since $\mathbf{v}_{is}$ is perpendicular to the  $\mathbf{v}_{ih}$, forming an encirclement and approaching the evader are two independent processes. A trade-off algorithm is needed to choose the values of surrounding coefficient $k_i$ and hunting coefficient $h_i$ to enable the capturing of the faster free-moving evader.

\subsection{Trade-off Algorithm}

In the pursuit-evasion game, 
the pursuers move with their maximum speed denoted by $V_i=\sqrt{||\mathbf{v}_{is}||_2^2+||\mathbf{v}_{ih}||_2^2}$. From  \eqref{sp1} and \eqref{34}, we have
\begin{align}\label{sp3}
    ||\mathbf{v}_{is}||_2= &k_ir_i|\epsilon_{i,i+1}-\epsilon_{i-1,i}|={V}_{i}\sin\beta_i, \\
    ||\mathbf{v}_{ih}||_2=&h_ir_i={V}_{i} \cos\beta_i, \label{spp3}
\end{align}
where $\beta_i \in [0,\frac{\pi}{2}]$ is the trade-off coefficient shown in Fig. \ref{fig6}.  Then, it yields
\begin{align}\label{kef}
k_i=&\frac{{V}_{i}\sin\beta_i}{r_i|\epsilon_{i,i+1}-\epsilon_{i-1,i}|}, \ \  \epsilon_{i,i+1}-\epsilon_{i-1,i} \not = 0, \ \ \beta_i \not= 0, \\
h_i=&\frac{{V}_{i}\cos\beta_i}{r_i}. \label{hef}
\end{align}

To ensure the property of \textbf{Theorem} \ref{t3}, $k_i>0$ is required. For the case of $|\epsilon_{i,i+1}-\epsilon_{i-1,i}|, \beta_i = 0$, the surrounding coefficient is set as $k_i=1$.   Note that $k_i$ and $h_i$ in \eqref{kef} and \eqref{hef} are time-varying parameters.

The value of surrounding coefficient $k_i$ and hunting coefficient $h_i$ are determined by the trade-off coefficient $\beta_i$ shown in \eqref{kef} and \eqref{hef}. 
Hence, the trade-off problem becomes how to design the trade-off coefficient $\beta_i$.  The trade-off algorithm for designing $\beta_i$ should have two characteristics:
\begin{enumerate}[(i)]
   \item $\beta_i(|\dot\alpha_i|,r_i)$ is a monotonic increasing function with respect to $|\dot\alpha_i|$;
   \item $\beta_i(|\dot\alpha_i|,r_i)$ is a concave function with respect to $r_i$.
\end{enumerate}

Next, the analysis for the above two characteristics is given:
\begin{enumerate}
\item
When $|\dot\alpha_i|=|k_i(\epsilon_{i,i+1}-\epsilon_{i-1,i})|$ is small, we have $\epsilon_{i,i+1} \rightarrow \epsilon_{i-1,i}$.  Small $|\dot\alpha_i|$ means pursuer $i$ 
tends not to change its polar angle, so pursuer $i$ should
focus on the hunting with small $\beta_i$. But when $|\dot\alpha_i|=|k_i(\epsilon_{i,i+1}-\epsilon_{i-1,i})|$ is large, the gap between coverage angles $\epsilon_{i,i+1}$ and $\epsilon_{i-1,i}$ is large.  
Pursuer $i$ should focus on the surrounding to decrease this gap with large $\beta_i$. Therefore, $\beta_i(|\dot\alpha_i|,r_i)$ should be designed as a monotonic increasing function with respect to $|\dot\alpha_i|$. 
From \eqref{8} and \eqref{13}, it yields
\begin{equation}\label{36}
|\dot\alpha_i| \le k_i \cdot (2\pi-\frac{\theta_{i+1}-\theta_{i-1}}{2}).
\end{equation}

\item
In the pursuit-evasion game, if a pursuer is close to the faster evader, it should focus on hunting because it has better chance to capture the evader. If a pursuer is far from the faster evader, it should also focus on hunting because it needs to increase its chance to capture the evader by decreasing its distance to the evader. Therefore, the $\beta_i(|\dot\alpha_i|,r_i)$ should be designed as a concave function with respect to $r_i$. 

Based on the above two characteristics, a trade-off algorithm for  $\beta_i$ is designed as
\begin{equation}\label{37}
\beta_i(|\dot\alpha_i|,r_i)=\frac{\pi}{2}(1-e^{-\delta_i\gamma_i}), 
\end{equation}
where surrounding factor $\delta_i \in [0,1]$ and hunting factor $\gamma_i \in [0,1]$ are
\begin{equation}\label{38}
\delta_i=\frac{2|\epsilon_{i,i+1}-\epsilon_{i-1,i}|}{4\pi-\theta_{i+1}+\theta_{i-1}}.
\end{equation}
\begin{equation}
\begin{array}{ll}
&\gamma_i=  sin(\pi(\frac{r_i}{r_{i}+ r_{i-1}+r_{i+1}})^{log_{3}^2}), \\
& r_{i}+ r_{i-1}+r_{i+1} \neq 0. \label{39}
\end{array}
\end{equation}
\end{enumerate}

\begin{remark}
{Note that the choice of trade-off coefficient $\beta_i$ is flexible. Here, we provide one solution \eqref{37}.}
\end{remark}

In a pursuit-evasion game, the aim is to capture the evader rather than purely surround the evader. Therefore, even if the surrounding factor $\delta_i$ equals its maximum value ($\max\delta_i=1$), the trade-off coefficient $\beta_i$ is still less than $\frac{\pi}{2}$. The pursuer $i$ will always has a component along the hunting direction, thus $h_i>0$ in \eqref{hef} is satisfied.

\section{Capture Condition}\label{discrete}

This section aims to find sufficient capture conditions, which guarantee that the group occupied angle $\theta_G$ equals $2\pi$ at any time instant $t$.

\begin{defn}\label{def1}
If the pursuer $\mathbf{p}_i$ can obtain the position information of the pursuer $\mathbf{p}_j$, pursuer $\mathbf{p}_j$ is called a neighbor of the pursuer $\mathbf{p}_i$. All the neighbors of the pursuer $\mathbf{p}_i$ are denoted by the set $\Omega_i$.
\end{defn}

\begin{defn}\label{def2}
A $n$-pursuer polygon shown in Fig. \ref{included} is a polygon formed by the $n$ pursuers, whose vertices are the pursuers' positions. Each pursuer $\mathbf{p}_i$ has two edges connecting to other two pursuers $\mathbf{p}_z$ and $\mathbf{p}_k$. The polygon maintenance set $S_{i}$ for the pursuer $\mathbf{p}_i$ is defined as $S_{i}=\{\mathbf{p}_z,\mathbf{p}_k\}$.
\end{defn}

\begin{figure}[!t]
\centering
\includegraphics[width=0.6\linewidth]{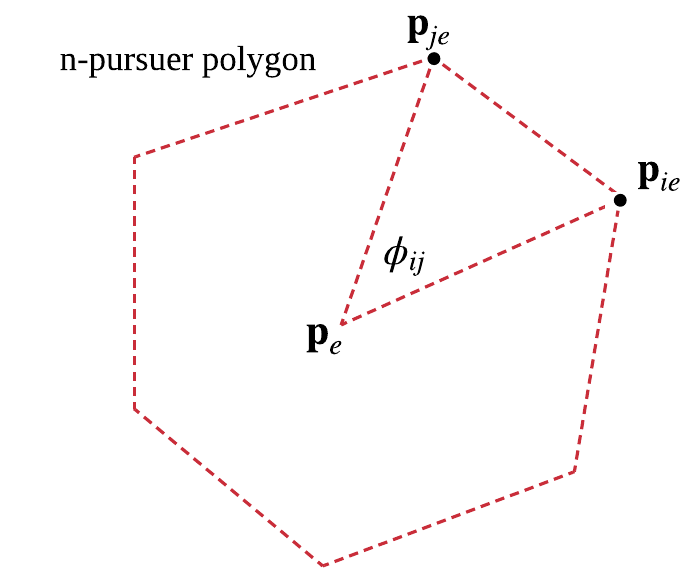}
\caption{$n$-pursuer polygon and included angle. }
\label{included}
\end{figure}

\begin{lemma}\label{lemma1}
In the local polar coordinate system $L$ centered at the evader' position, the included angle $\phi_{ij}\le \pi$ formed by two pursuers and the evader shown in Fig. \ref{included} satisfies $\phi_{ij}- \frac{\theta_{i}+\theta_{j}}{2} \le 0$ if the distance $d_{i,j}$ between the pursuers $\mathbf{p}_{ie}(\alpha_i,r_i)$ and $\mathbf{p}_{je}(\alpha_{j},r_{j})$ satisfies
\begin{equation}\label{dis}
    d_{i,j} \le 2d_c\min (\lambda_i,\lambda_{j}).
\end{equation}
\end{lemma}
\begin{proof}
From \eqref{4}, it yields $\cos \theta_i= 1 - 2\lambda_i^2$ and $\theta_i<\pi$. Then, we have $r_i^2+r_{j}^2-2r_ir_{j}\cos \theta_i = (r_i-r_{j})^2 + 4r_ir_{j}\lambda_i^2
\ge 4r_ir_{j}\lambda_i^2 \ge 4d_c^2\lambda_i^2$. For the included angle $\phi_{ij}$:
\begin{equation}
\begin{array}{ll}
    \cos \phi_{ij} &= \frac{r_i^2+r_{j}^2-d_{i,j}^2}{2r_ir_{j}}  \\ &\ge \frac{r_i^2+r_{j}^2-4d_c^2\lambda_i^2}{2r_ir_{j}} \\
    & \ge \frac{r_i^2+r_{j}^2-(r_i^2+r_{j}^2-2r_ir_{j}\cos \theta_i)}{2r_ir_{j}}=\cos \theta_i.
\end{array}
\end{equation}
Therefore, $\phi_{ij} \le  \theta_i$. Similarly,  $\phi_{ij} \le  \theta_{j}$. It is easy to know the coverage angle $\epsilon_{i,j}=\phi_{ij}- \frac{\theta_{i}+\theta_{j}}{2} \le 0$. 
\end{proof}

\begin{lemma}\label{lemma2}
For a faster evader $\mathbf{p}_e$ which is inside a $n$-pursuer polygon, if the edge length of the polygon is no more than $2d_c\lambda_{\min}$, where $\lambda_{\min} = \min(\lambda_1,\cdots,\lambda_n)$ is the minimum speed ratio among all the pursuers, the group occupied angle $\theta_G$ equals $2\pi$.
\end{lemma}
\begin{proof}

The distance between the faster evader $\mathbf{p}_e$ and any pursuer is larger than $d_c$ if the evader is not captured. For the local polar coordinate system $L$ centered at the evader' position $\mathbf{p}_e$,  each edge in the $n$-pursuer polygon has two vertices where two pursuers $\mathbf{p}_{ie}(\alpha_i,r_i)$ and $\mathbf{p}_{je}(\alpha_{j},r_{j})$ are situated. Since the distance $d_{i,j}$ between $\mathbf{p}_{ie}$ and $\mathbf{p}_{je}$ satisfies $d_{i,j} \le 2d_c \lambda_{\min}$, the coverage angle $\epsilon_{i,j}$ satisfies $\epsilon_{i,j}=\phi_{ij} - \frac{\theta_i+\theta_{j}}{2}  \le 0$ from Lemma \ref{lemma1}. Therefore, all the coverage angles $\epsilon_{i,j}$  are no more than $0$, then the group occupied angle $\theta_G$ equals $2\pi$ from \eqref{8} and \eqref{9}.
\end{proof}

To ensure the length of any edge in the $n$-pursuer polygon is no more than $2d_c\lambda_{\min}$, a distance maintenance set $D_i^m$ for the pursuer $\mathbf{p}_i$ is defined as
\begin{equation}
    D_i^m = \{ \mathbf{p}_j :  R_c \le ||\mathbf{p}_i-\mathbf{p}_j||_2 < R_f=2d_c\lambda_{\min}, \mathbf{p}_j \in S_i \},  
\end{equation}
where $S_i$ is the polygon maintenance set defined in the \textbf{Definition \ref{def2}}.  The following potential function is adopted \cite{han2015formation} to maintain the distance between $\mathbf{p}_i$ and any pursuer $\mathbf{p}_j \in D_i^m$
\begin{equation}\label{distance}
Q_{ij} = \left\{\begin{array}{lll} \left(\frac{||\mathbf{p}_i-\mathbf{p}_j||_2^2-(2R_c^2-R_f^2)}{||\mathbf{p}_i-\mathbf{p}_j||_2^2-R_f^2}\right)^2-1,  \mathbf{p}_j \in D_i^m, \\ 0, \ \ \text{otherwise},
\end{array}\right.
\end{equation}
where $R_f>R_c>0$. $\frac{\partial Q_{ij}}{\partial \mathbf{P}_i}=0$ if $\mathbf{p}_j \not \in D_i^m$, and $\frac{\partial Q_{ij}}{\partial \mathbf{P}_i}\not=0$ if $\mathbf{p}_j  \in D_i^m$.

\begin{figure}[!t]
\centering
\includegraphics[width=0.8\linewidth]{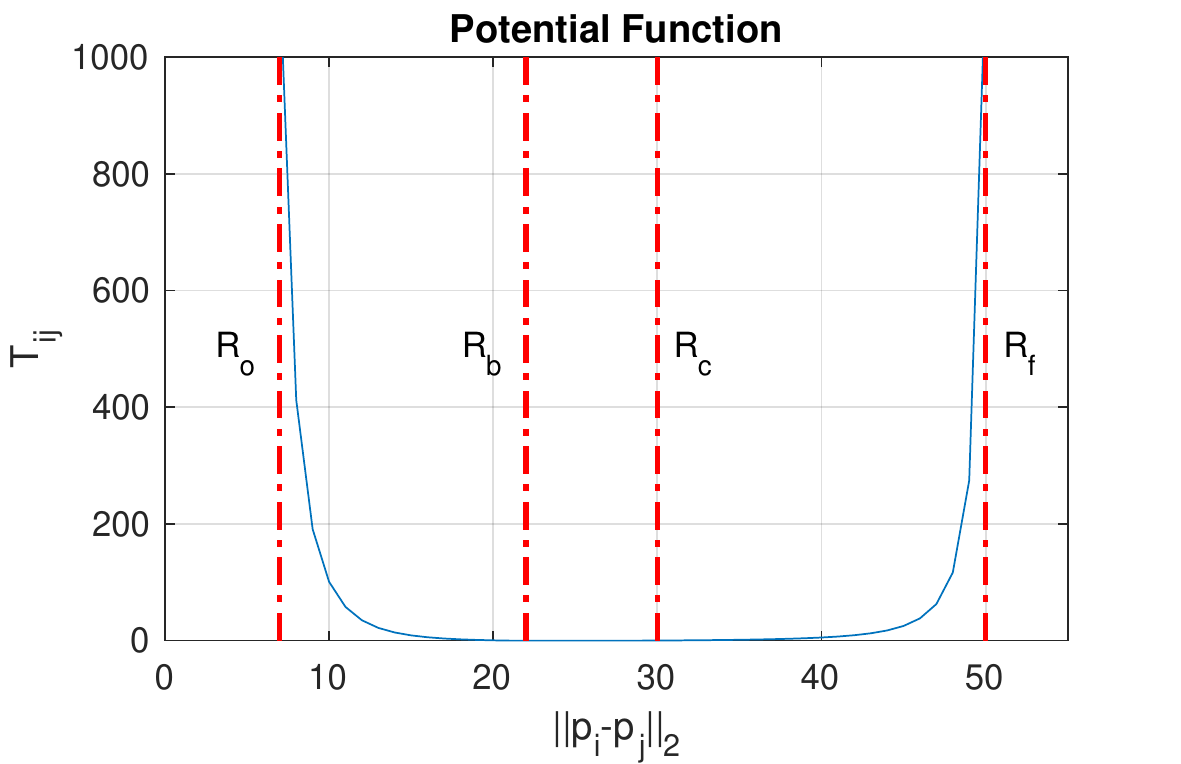}
\caption{Potential function of inter-collision avoidance and distance maintenance. }
\label{potential1}
\end{figure}

Let
\begin{equation}
    \mathbf{s}_i = -\sum\limits_{\mathbf{p}_j \in S_i} \frac{\partial Q_{ij}}{\partial \mathbf{p}_i}.
\end{equation}

If the pursuers are required to keep a specified safe distance $R_o$ from each other, an inter-collision avoidance set is defined for the pursuer $\mathbf{p}_i$
\begin{equation}
        D_i^o = \{ \mathbf{p}_j:  R_o < ||\mathbf{p}_i-\mathbf{p}_j||_2 \le R_b, \mathbf{p}_j \in \Omega_i \},  
\end{equation}
where $\Omega_i$ is defined in the \textbf{Definition \ref{def1}}. The following potential function can be used to avoid inter-collision
\begin{equation}\label{61}
    U_{ij} = \left\{\begin{array}{lll} \left(\frac{||\mathbf{p}_i-\mathbf{p}_j||_2^2-(2R_b^2-R_o^2)}{||\mathbf{p}_i-\mathbf{p}_j||_2^2-R_o^2}\right)^2-1,  \mathbf{p}_j \in D_i^o, \\ 0, \ \ \text{otherwise},
\end{array}\right.
\end{equation}
where $R_b>R_o>0$.  $\frac{\partial U_{ij}}{\partial \mathbf{P}_i}=0$ if $\mathbf{p}_j \not \in D_i^o$, and $\frac{\partial U_{ij}}{\partial \mathbf{P}_i}\not=0$ if $\mathbf{p}_j  \in D_i^o$. 
Let
\begin{equation}
    \mathbf{w}_i = -\sum\limits_{\mathbf{p}_j \in\Omega_i } \frac{\partial U_{ij}}{\partial \mathbf{p}_i}.
\end{equation}

The distance maintenance and inter-collision avoidance velocity control $\mathbf{v}_{im}$ is designed as
\begin{equation}
    \mathbf{v}_{im} = (||\mathbf{v}_{is}+ \mathbf{v}_{ih}||_2 + b)sgn(\mathbf{s}_i+\mathbf{w}_i),
\end{equation}
where $sgn(\cdot)$ is the signum function and $b>0$. Then, the velocity control law $\mathbf{v}_i$ in \eqref{sp2} becomes
\begin{equation}\label{sp41}
    \mathbf{v}_i = \mathbf{v}_{is} +\mathbf{v}_{ih} +\mathbf{v}_{im}.
\end{equation}

Considering that pursuer $i$ moves with maximum speed denoted by $V_i$, the velocity control law \eqref{sp41} can be revised as
\begin{equation}\label{sp5}
       \mathbf{v}_i = V_i \frac{\mathbf{v}_{is} +\mathbf{v}_{ih} +\mathbf{v}_{im}}{||\mathbf{v}_{is} +\mathbf{v}_{ih} +\mathbf{v}_{im}||_2}. 
\end{equation}

For the potential function $T_{ij}=Q_{ij}+U_{ij}$, an example is given in Fig. \ref{potential1} with $R_c=30$, $R_f=50$, $R_b=20$, and $R_o=5$.

\begin{theorem}\label{pp1}
In a pursuit-evasion game with a faster free-moving evader who is inside a $n$-pursuer polygon at initial time instant $t_0$, the velocity control law $\eqref{sp5}$ achieves the desired cooperative pursuit with distance maintenance and inter-collision avoidance at time instant $t>t_0$, such that the faster free-moving evader will be captured by the pursuers if the following conditions are satisfied:
\begin{enumerate}
    \item $||\mathbf{p}_{i}(t_0)-\mathbf{p}_{j}(t_0)||_2 < R_f, 
    i=1,2,\cdots,n,\mathbf{p}_j \in S_i$;
    \item 
    $R_o<||\mathbf{p}_{i}(t_0)-\mathbf{p}_{j}(t_0)||_2, i=1,2,\cdots,n, \mathbf{p}_j \in \Omega_i$;
    \item
    $R_o<R_b<R_c<R_f \le 2d_c\lambda_{\min}$;
    \item
     $\{\mathbf{p}_{i-1},\mathbf{p}_{i+1}\} \in S_i \in \Omega_i$.
\end{enumerate}

\begin{figure}[t]
\centering
\includegraphics[width=0.8\linewidth]{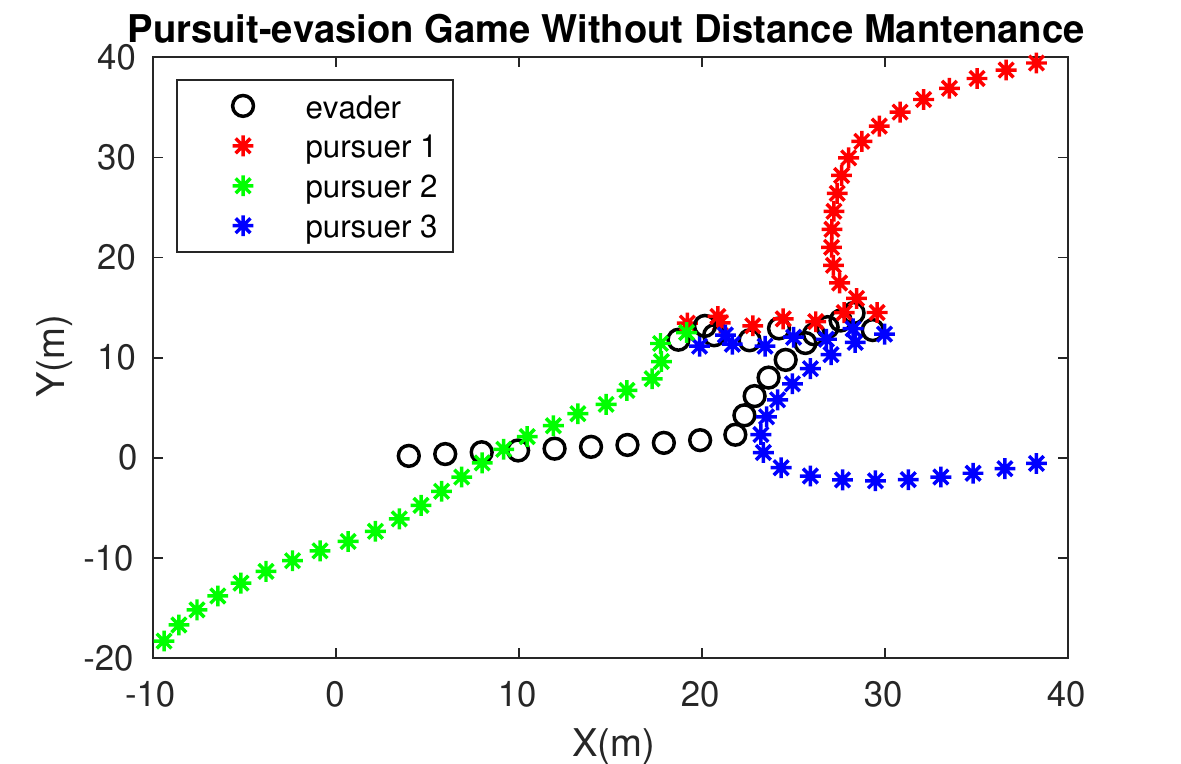}
\caption{Pursuit-evasion game with speed ratio of pursuer-evader being 0.9.}
\label{case1}
\end{figure}

\begin{proof}

Since $R_b <R_c$, the potential functions $Q_{ij}$ and $U_{ij}$ work in the different regions and do not affect each other. Therefore, the intersection of distance maintenance set $D_i^m$ and inter-collision avoidance set $D_i^o$ is null ($D_i^m \cap D_i^o = \emptyset $).  Since $\mathbf{s}_i+\mathbf{w}_i$ has the same sign with $\mathbf{v}_{is}+\mathbf{v}_{ih}+(||\mathbf{v}_{is}+\mathbf{v}_{ih}||_2+b)sgn(\mathbf{s}_{i}+\mathbf{w}_{i})$, for the pursuer $i$, we have 
\begin{equation}
\begin{array}{ll}
 & \sum\limits_{p_j \in \Omega_i}(\frac{\partial Q_{ij}}{\partial \mathbf{p}_i} +\frac{\partial U_{ij}}{\partial \mathbf{p}_i})  \mathbf{\dot p}_i  \\ 
 & = - (\mathbf{s}_i + \mathbf{w}_i) \cdot \mathbf{ \dot p}_i \\
 & = -(\mathbf{s}_i + \mathbf{w}_i) \cdot \frac{V_i (\mathbf{v}_{is} +\mathbf{v}_{ih} +\mathbf{v}_{im})}{||\mathbf{v}_{is} +\mathbf{v}_{ih} +\mathbf{v}_{im}||_2}  \\
 & =  -\frac{V_i(\mathbf{s}_i + \mathbf{w}_i)}{||\mathbf{v}_{is} +\mathbf{v}_{ih} +\mathbf{v}_{im}||_2}(\mathbf{v}_{is} +\mathbf{v}_{ih} + \\
 & (||\mathbf{v}_{is}+\mathbf{v}_{ih}||_2+b)sgn(\mathbf{s}_{i}+\mathbf{w}_{i})) \le 0.
\end{array}
\end{equation}

Considering a Lyaponov function candidate $V=\sum\limits_{i=1}^n\sum\limits_{p_j \in \Omega_i}(Q_{ij}+U_{ij})$, it yields
\begin{equation}
\begin{array}{ll}
\dot V & = \sum\limits_{i=1}^n\sum\limits_{p_j \in \Omega_i} (\frac{\partial Q_{ij}}{\partial \mathbf{p}_i} +\frac{\partial U_{ij}}{\partial \mathbf{p}_i}) \mathbf{ \dot p}_i  \le 0.
\end{array}
\end{equation}

The distances $||\mathbf{p}_{i}(t)-\mathbf{p}_{j}(t)||_2< R_f \le 2d_c\lambda_{\min},  i=1,2,\cdots,n,\mathbf{p}_j \in S_i$ and $R_o<||\mathbf{p}_{i}(t)-\mathbf{p}_{j}(t)||_2,   i=1,2,\cdots,n, \mathbf{p}_j \in \Omega_i$
are guaranteed at time instant $t>t_0$ because if  $||\mathbf{p}_{i}(t)-\mathbf{p}_{j}(t)||_2 \rightarrow R_f$ and $||\mathbf{p}_{i}(t)-\mathbf{p}_{j}(t)||_2 \rightarrow R_o$, the Lyapunov function will go to infinity, which contradicts the fact that $\dot V \le 0$.  Therefore, the pursuers can avoid the inter-collision and
the group occupied angle $\theta_G$ equals $2\pi$ at any time instant $t > t_0$ from Lemma \ref{lemma2}.
In other words, the faster evader will be trapped in the union of each pursuer' capture domain at any time instant $t > t_0$ and will be finally captured regardless of the strategy the evader adopts. 
\end{proof}
\end{theorem}

It is well known that when the length of any side of a convex polygon is no more than $2d_c\lambda_{\min}$, the radius of the circumcircle is no more than $\frac{d_c\lambda_{\min}}{\sin (\frac{\pi}{n})}$. 
For forming a convex polygon with the radius of the circumcircle being $R_p>d_c$, the number of the pursuers $n$ required to ensure capture satisfies:
\begin{equation}\label{number}
   \frac{d_c\lambda_{\min}}{\sin (\frac{\pi}{n})}> R_p \Longrightarrow n > \frac{\pi}{\arcsin{(\frac{d_c\lambda_{\min}}{R_p}})}.
\end{equation}

Compared with the existing sufficient capture conditions in \cite{Xu:2007wu} that the pursuers should have the same maximum speed and should be angle-evenly distributed around the evader at the beginning of the game, the proposed capture conditions only require the distance maintenance of the $n$-pursuer polygon.

\section{Simulation}\label{simulation}

\subsection{Pursuit-evasion Game Simulation without Distance Maintenance}
The three-pursuer one-evader game is simulated using Matlab. We first show the performance of distributed pursuit algorithm \eqref{sp2} with the trade-off algorithm \eqref{37}. 
In order to occupy all the free-moving directions of the evader, we have $\theta_i \ge \frac{2\pi}{3}$ from \eqref{11}. Then the speed ratio must satisfy $\lambda_i\ge sin(\frac{\theta}{2})=0.86$ from \eqref{4}. In the simulation, the speed ratio of pursuer $i$ to the evader is set as $\lambda_i=\frac{V_{i}}{V_e}=\frac{9}{10}$ with $V_e= 2,i=1,2,3$. The initial positions of the pursuers and the evader are set as $(40, 40), (-10, -20), (40, 0)$, and $(2, 0)$.  The terminal condition is set as $d_c=1$. The neighbors of pursuer $i$ are $\Omega_i=\{\mathbf{p}_j:||\mathbf{p}_j-\mathbf{p}_i|| \le 100,j \not = i, j=1,2,3 \}$.

\begin{figure}[t]
\centering
\includegraphics[width=0.8\linewidth]{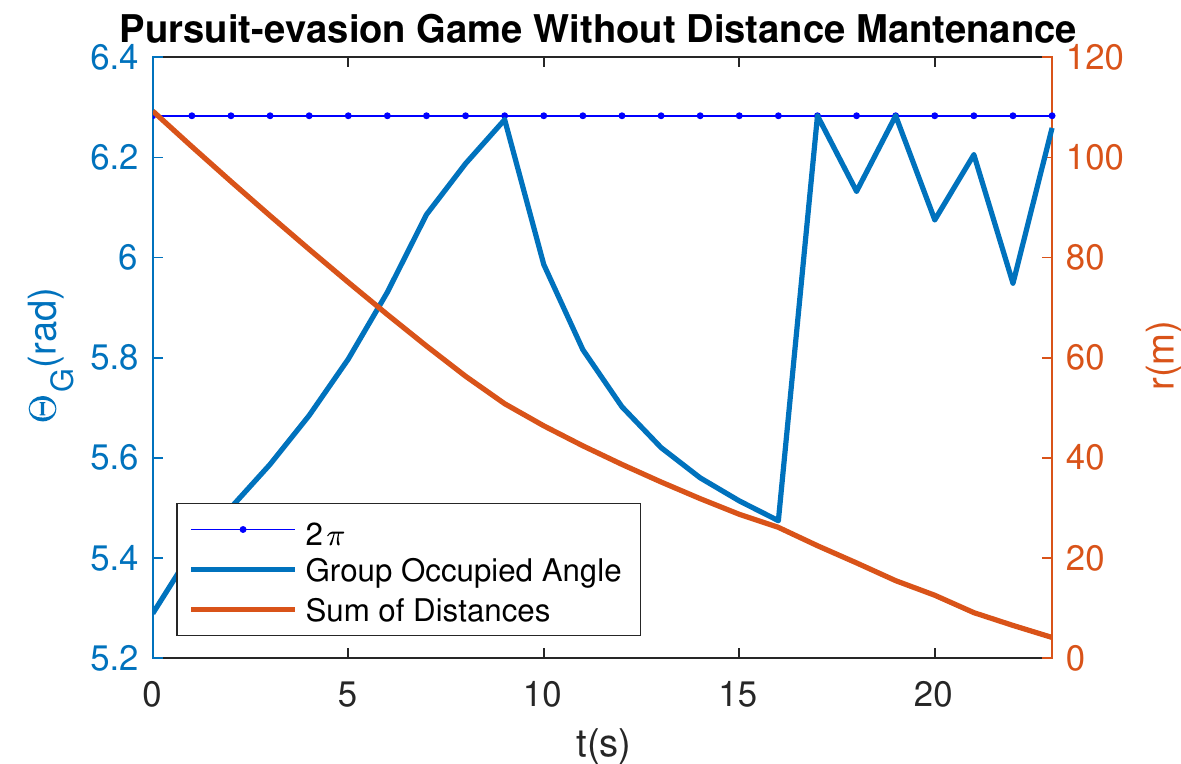}
\caption{Group occupied angle and sum of distances to the evader.}
\label{case2}
\end{figure}

\subsubsection{Strategy of the Evader}

The strategy in \eqref{eq:evader-flee} makes the evader react more responsively to the nearer pursuers.
\begin{equation}\label{eq:evader-flee}
\mathbf{v}_e =  V_e \frac{\sum_i \frac{k_e(\mathbf{p}_e-\mathbf{p}_i
)}{r_i}+\mathbf{p}_e}{||\sum_i \frac{k_e(\mathbf{p}_e-\mathbf{p}_i
)}{r_i}+\mathbf{p}_e||_2},
\end{equation}
where the parameter in \eqref{eq:evader-flee} is set as $k_e=140$. The simulation results are shown in Fig. \ref{case1} and Fig. \ref{case2}, respectively. The pursuers first tried to form an encirclement because there is an escapable angle for the evader to escape. The group occupied angle is increased to $2\pi$ at time instant $9s$, then the pursuers focuses on the hunting  and the evader is finally captured. The simulation result validates the effectiveness of distributed pursuit algorithm \eqref{sp2} with the trade-off algorithm \eqref{37}. 

\begin{figure}[!t]
\centering
\includegraphics[width=0.8\linewidth]{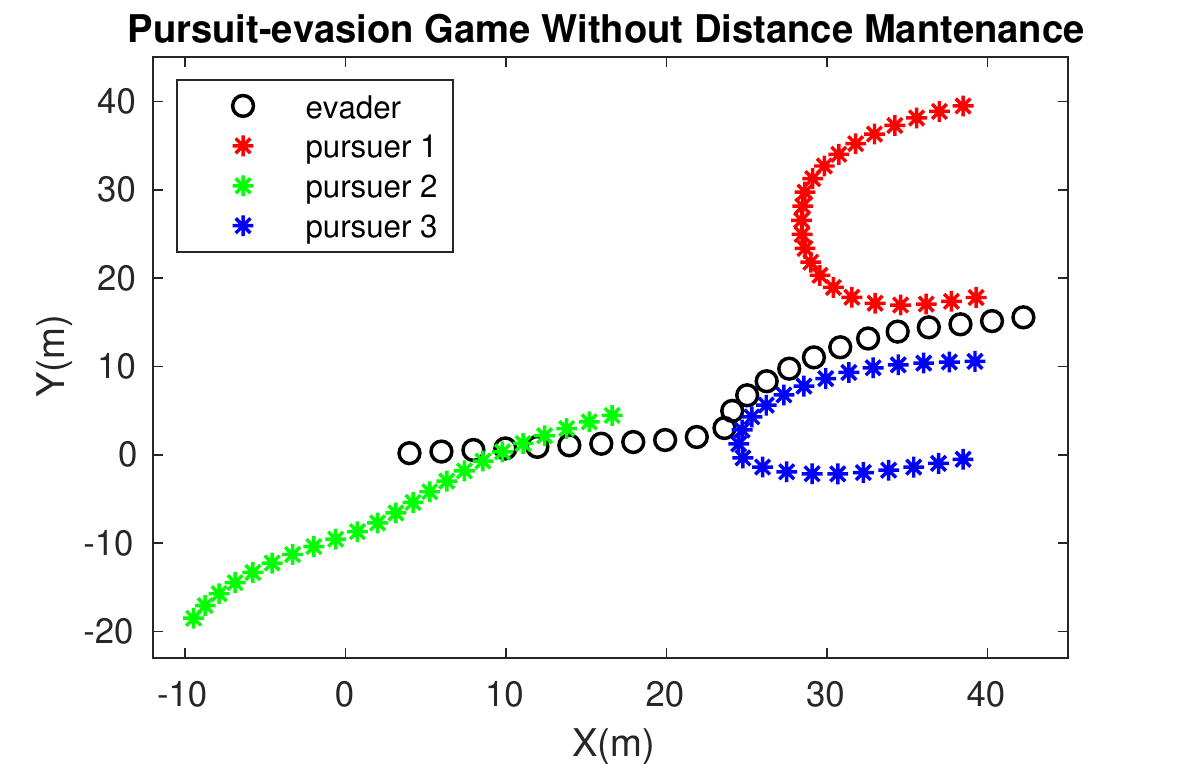}
\caption{Pursuit-evasion game with speed ratio of pursuer-evader being 0.8.}
\label{case5}
\end{figure}

\begin{figure*}[!t]
		\centering
		\begin{subfigure}[t]{0.32\textwidth}
		\centering
		\includegraphics[width=1\linewidth]{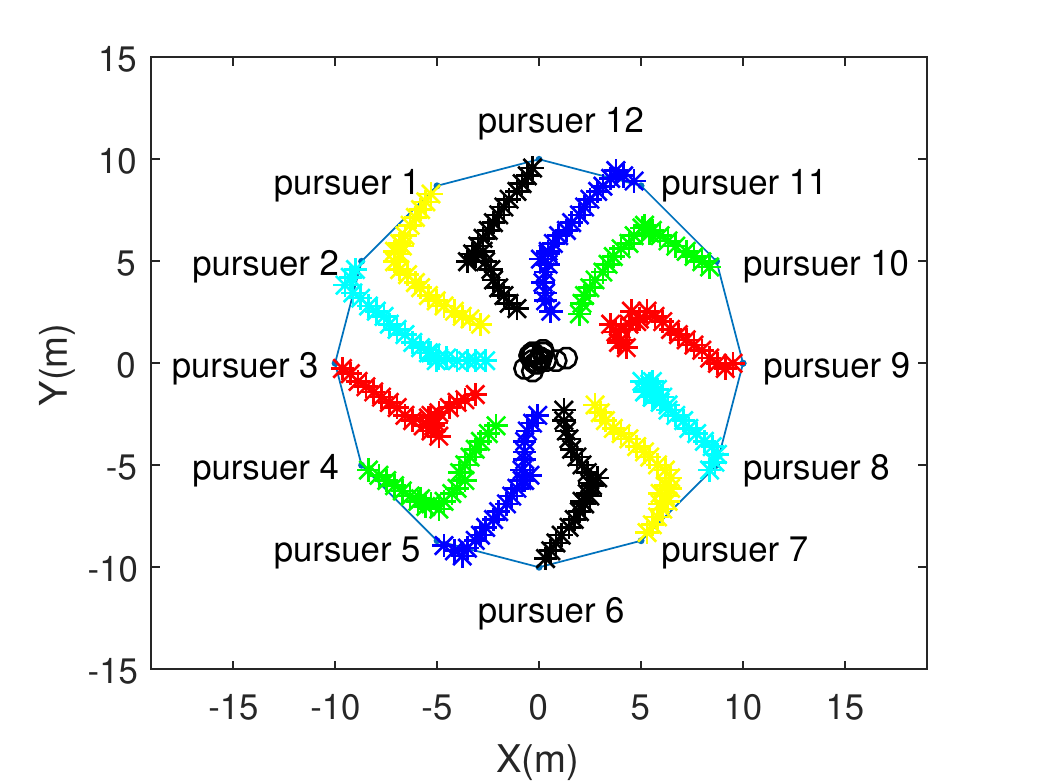}
		\caption{The trajectories of the game.}
		\label{case7}
	\end{subfigure}
	\begin{subfigure}[t]{0.32\textwidth} 
		\centering
		\includegraphics[width=1\linewidth]{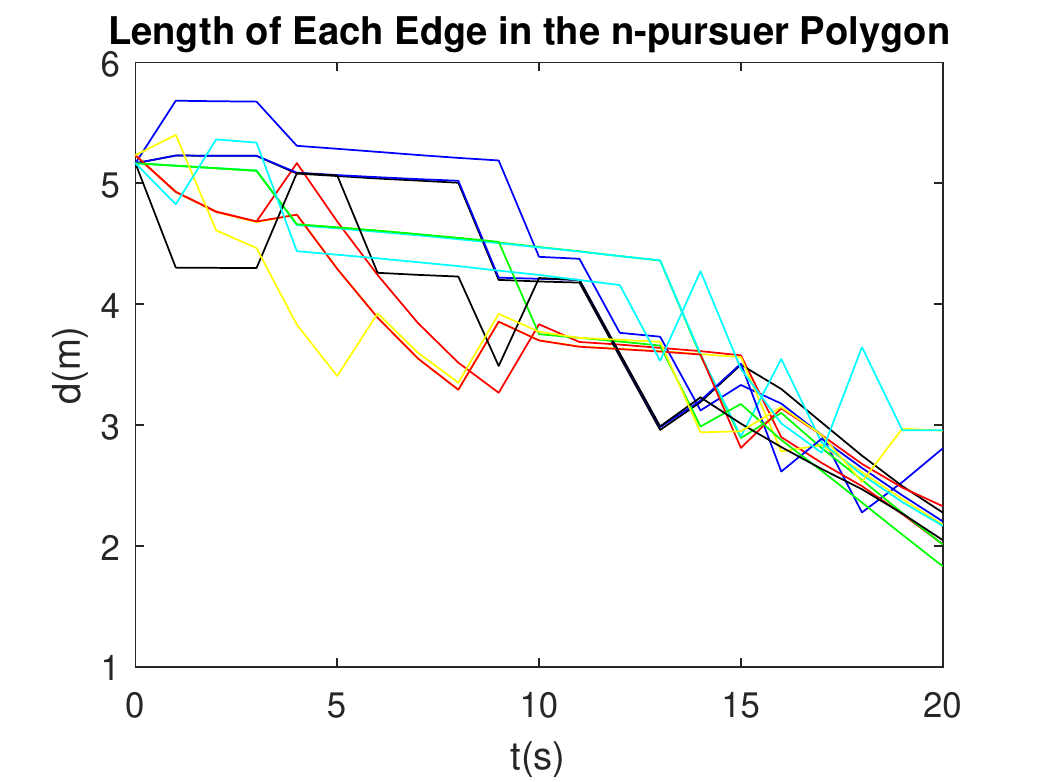}
		\caption{Length of edges in the $n$-pursuer polygon.}
		\label{case9}
	\end{subfigure}
	\begin{subfigure}[t]{0.32\textwidth}
	\centering
	\includegraphics[width=1\linewidth]{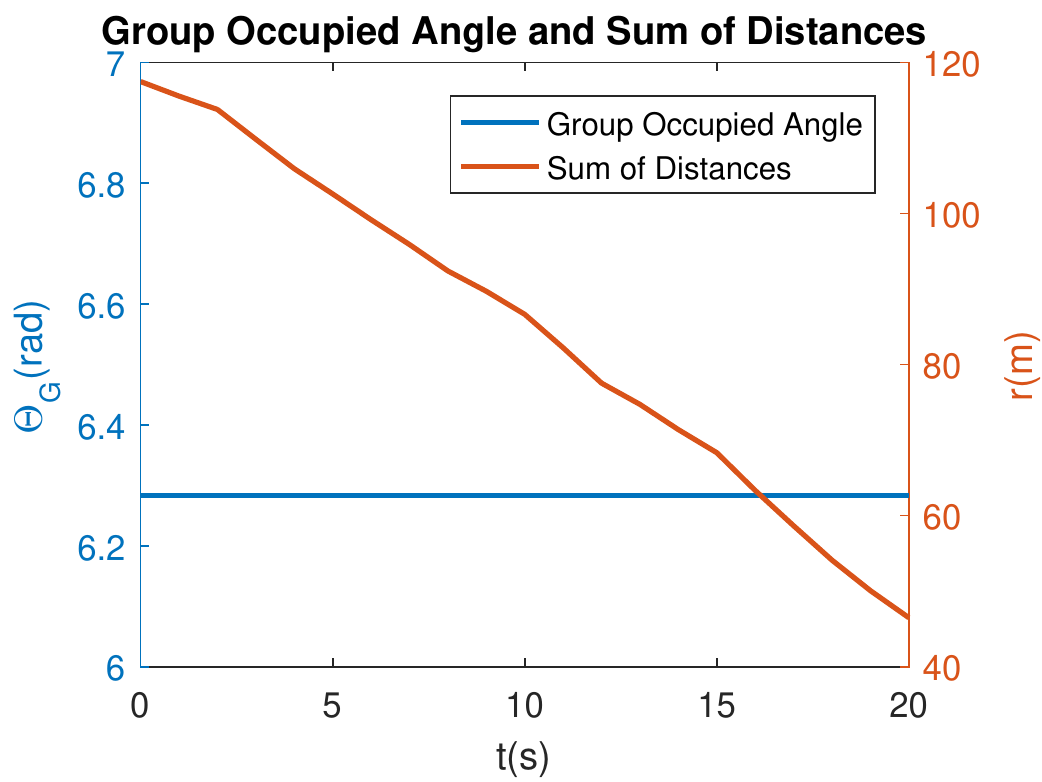}
	\caption{The occupied angle and sum of distances.}
	\label{case8}
	\end{subfigure}
	\caption{The pursuit-evasion game with distance maintenance and inter-collision avoidance.}\label{fig:pursuit-evasion-simulation}
\end{figure*}

\subsubsection{The effects of Speed Ratio on Pursuit-evasion Game}

One important parameter for pursuit-evasion game is the speed ratio. 
Smaller speed ratio $\frac{V_{i}}{V_e}$ means that pursuer $i$ has less chance to capture the evader.  If the speed ratio is set as $\lambda_i=\frac{V_{i}}{V_e}=0.8$  with $V_e= 2, i=1,2,3$, the occupied angle is $\theta_i=2\arcsin(0.8)\simeq0.59\pi<\frac{2\pi}{3},i=1,2,3$. Therefore, the conditions $\epsilon_{1,2}(t),\epsilon_{2,3}(t),\epsilon_{3,1}(t)<0$ can not be satisfied at any time instant. There is always an escapable angle for the evader to escape from capture shown in Fig. \ref{case5}.

\subsection{Pursuit-evasion Game Simulation with Distance Maintenance and Inter-collision Avoidance}

The terminal condition is set as $d_c=3.1$. The speed ratio is set as $\lambda_i =\frac{V_{i}}{V_e}= 0.95$ with $V_e= 0.5,i=1,\cdots,12$. The parameters in \eqref{distance} and \eqref{61} are chosen as $R_c=3.5$,
$R_f=5.7 < 2d_c\lambda_i$, $R_o=0.5$, $R_b=3$, and $b=1$. The initial positions of twelve-pursuer polygon are set as $(10,0),(8.7,5),(5,8.7),(0,10),(-5,8.7),(-8.7,5),(-10,0)$, $(-8.7,-5),(-5,-8.7),(0,-10),(5,-8.7),(8.7,-5)$. The other parameters are set as $\Omega_i=\{\mathbf{p}_j:||\mathbf{p}_j-\mathbf{p}_i|| \le 50,j \not = i\}$ and $S_i=\{\mathbf{p}_{i-1}(t_0),\mathbf{p}_{i+1}(t_0)\}$. The initial distribution satisfies $||\mathbf{p}_{i}(t_0)-\mathbf{p}_{j}(t_0)||_2 < R_f, \mathbf{p}_j \in S_i$ and $R_o<||\mathbf{p}_{i}(t_0)-\mathbf{p}_{j}(t_0)||_2, \mathbf{p}_j \in \Omega_i$. 
From \textbf{Theorem} \ref{pp1}, the evader will be captured under the velocity control law \eqref{sp5} regardless of the strategy of the evader.
The simulation results are presented in Fig. \ref{fig:pursuit-evasion-simulation}. The length of each edge is maintained $||\mathbf{p}_{i}(t)-\mathbf{p}_{j}(t)||_2 < R_f < 2d_c\lambda_{\min} (i=1,\cdots,12;\mathbf{p}_j \in S_i;t>t_0)$. The group occupied angle equals $2\pi$ at any time instant $t>t_0$. There is no chance for the evader to escape from capture, which will be finally captured.

\section{Experiments}\label{experiment}

We will demonstrate the pursuit-evasion game using multiple ground robots. 
More information about the experiments is released at 
\url{https://chenwang.site/cooperative-pursuit}. The kinematic models of the ground robot are more complex than the proposed kinematic model \eqref{kinetic}. 
The proposed velocity control law $\mathbf{v}_i$ can be regarded as the cooperative path planning for robots.

We built four micro-ground robots, each of which contains an embedded computing board for decision making and a ZigBee module for communication with an external camera positioning system, which is able to recognize and localize the colored circle markers carried by robots.
The camera positioning system is implemented in Ubuntu using the OpenCV library. The speed ratio is set as $\lambda_i=\sin (\frac{3\pi}{8}),i=1,2,3$.  The experiment is carried out under the velocity control law \eqref{sp2} with the trade-off algorithm \eqref{37}.
Since the moving direction and escape strategy of the evader are unknown to the pursuers, to test the unpredictable moving direction, we choose to control the evader manually.
Some snapshots of the pursuit-evasion game from the overlook view are shown in Fig. \ref{fig:pursuit-evasion}, where the real-time moving directions of the robots are marked by the arrows.  
The experimental results show that the pursuers can achieve capture by cooperation even if the moving direction and escape strategy of the evader are unknown to the pursuers.
The evader slows down when it makes a turn and then the pursuers get an opportunity to capture it by cooperation.

\begin{figure*}[!t]
	\centering
	\includegraphics[width=1\linewidth]{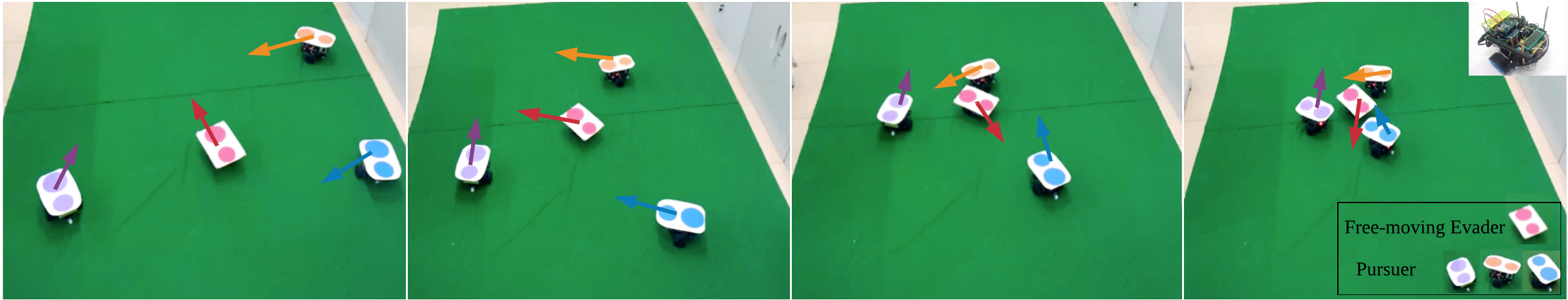}
	\centering
	\caption{The pursuit-evasion experiment using multi-ground-robots. }
	\label{fig:pursuit-evasion}
\end{figure*}

\section{Conclusion}\label{conclusion}

In this paper, a  distributed pursuit algorithm is proposed for the pursuit-evasion game with a faster free-moving evader. The encirclement algorithm can form and keep an encirclement, which traps and keeps the faster evader into the union of each pursuer' capture domain. The trade-off algorithm can balance between forming an encirclement and approaching the evader. In addition, sufficient capture conditions are derived that involve the initial spatial distribution and speed ratios to guarantee capture.

\section{Acknowledgement}
We would like to thank Mr. Huiyang Deng, Chengsi Shang, Haoqing Lan, and Xiaofei Long for their help in the pursuit-evasion experiment. We also appreciate the help from Mr. Muqing Cao, Thien-Minh Nguyen, and Hoang-Thien Nguyen for their suggestions.

\ifCLASSOPTIONcaptionsoff
  \newpage
\fi


\bibliographystyle{IEEEtran}

\bibliography{papers}

\end{document}